\documentclass{article}

\usepackage{arxiv}

\usepackage[utf8]{inputenc} % allow utf-8 input
\usepackage[T1]{fontenc}    % use 8-bit T1 fonts
\usepackage{hyperref}       % hyperlinks
\usepackage{url}            % simple URL typesetting
\usepackage{booktabs}       % professional-quality tables
\usepackage{amsfonts}       % blackboard math symbols
\usepackage{amsmath} 
\usepackage{amsthm} 
\usepackage{amssymb}
\usepackage{nicefrac}       % compact symbols for 1/2, etc.
\usepackage{microtype}      % microtypography
\usepackage{lipsum}
\usepackage{graphicx}
\graphicspath{ {./Img/} }

\theoremstyle{plain}
\newtheorem{theorem}{Theorem}

\newtheorem{corollary}{Corollary}

\newtheorem{statement}{Statement}

\newtheorem{conjecture}{Conjecture}

\theoremstyle{definition}
\newtheorem{definition}{Definition}
\newtheorem{remark}{Remark}
\newtheorem{example}{Example}

\title{An alternative way of defining finite graphs}

\author{
 Maxim Nazarov\\
  Chair of Higher Math 1\\
  Moscow Institute of Electronic Technology\\
  \texttt{Nazarov-Maximilian@yandex.ru} \\
}

\begin{document}
\maketitle
\begin{abstract}
In this paper we introduce "graph linear notation" --- a complete graph invariant --- which is positioned as an alternative definition for the finite graphs. This invariant is constructed using an algorithm similar to the algorithm of finding canonical forms of graphs. 
Storing graph linear notation instead of a regular graph allows us to greatly simplify  two major problems: the construction of illustrations for graphs with regards to possible graph symmetries, and the comparison of two graphs for isomorphism. We also demonstrate the  transferability to the graph linear notations such classical graph theory concepts as colourings and graph paths.
\end{abstract}

% keywords can be removed

\section{Introduction}
In the article \cite{Nazarov_Graph1}  the linear notation $I[G]$ for the abstract graphs was defined. It was demonstrated that this $I[G]$ notation is a complete graph invariant and has properties similar to those of ordinary graphs $G$. 
In particular, for $I[G]$ the concepts of abstract vertices and abstract edges can be introduced, as well as subgraphs, colourings, paths, and some elementary operations on abstract graphs. 
It should be noted that the algorithm for transitioning from $I[G]$ to ordinary graphs $G$ is polynomial in complexity. Given that $I[G]$ is a complete graph invariant, this polynomial complexity justifies using the $I[G]$ notation as an alternative way to store graphs in databases. 

The linear notation $I[G]$ was implemented using the numbers $I(\overline{v})$ and $I(\overline{u,v})$ of the automorphism classes of the vertices and edges of graph $G$. 
For these numbers, it was proved that $I(\overline{v})$ and $I(\overline{u,v})$ are unique identifiers for automorphism classes of vertices and edges. This fact allows us to introduce linear ordering on the set of these classes: $\overline{v_{1}} \leq  \overline{v_{2}} \Leftrightarrow I(\overline{v_{1}}) \leq  I(\overline{v_{2}})$ and $(\overline{v_{1},u_{1}}) \leq  (\overline{v_{2},u_{2}}) \Leftrightarrow I(\overline{v_{1},u_{1}}) \leq  I(\overline{v_{2},u_{2}})$.
This means that we can number the vertex classes as $1,...,m$ and the edge classes as $1,...,k$ if the graph has $m$  vertex automorphism classes and $k$ edge automorphism classes. It should be noted that in the paper \cite{Nazarov_Graph1} there was an inaccuracy and in all the illustrations and examples the vertices and edges were labelled with precisely these ordinal indices, and not $I(\overline{v})$ and $I(\overline{u,v})$ as was stated in the captions.
In the current work, we will first distinguish between indices and numbers of automorphism classes of vertices and edges and reformulate all the main results of the article \cite{Nazarov_Graph1} for  the indices of automorphism classes.

To demonstrate the advantages of the linear notations $I[G]$ over the classical description, we will additionally consider examples of defining standard graphs, and also present algorithms for constructing paths and implementing colourings for the linear notations $I[G]$.

\section{Constructing the linear notations for classes of isomorphic graphs}
\label{sec:Constructing}
We exclude from consideration all alternative graphs, such as multiple, directed, and infinite-loop graphs (see \cite{Zykov1987}). We note that generalizing our method will be relatively straightforward only for classes of directed graphs and classes of graphs with loops.

\begin{definition}
	\textit{A graph} is a pair $G = (V,E)$, where the set of vertices $V$ --- is any finite set, and the set of edges $E\subseteq V\times V $ --- is a binary relation on $V$ for which the following conditions are satisfied:
\begin{enumerate}
	\item $\forall a,b \in V \quad ((a,b) \in E \Rightarrow (b,a) \in E$) is symmetric; 
	\item $\forall a\in V \quad \quad ((a,a) \notin E)$ is anti-reflexive. 
\end{enumerate}
\end{definition} 

\begin{definition}
	Let $G= (V,E)$ be an arbitrary graph such that $|V|=n$. Then, if we choose some order for the set of vertices $\alpha = (v_{1},...,v_{n})$, then we can associate $G$ with an adjacency matrix $A$ according to the following rule: 
\begin{center}
$
	A({i},{j}) = 1 \Leftrightarrow (v_{i},v_{j}) \in E 
	\quad \wedge \quad
	A({i},j) = 0 \Leftrightarrow (v_{i},v_{j}) \notin E
$
\end{center}
\end{definition}

Figure \ref{nazarov:gr_table} shows an example of two adjacency matrices that were constructed for different orders $\alpha_{1}$ and $\alpha_{2}$ on the vertex set $V$ of the graph $G$.

\begin{figure}[htb] 
  \vspace{-6pt}
  \begin{center}
     \includegraphics[width=0.7\textwidth]{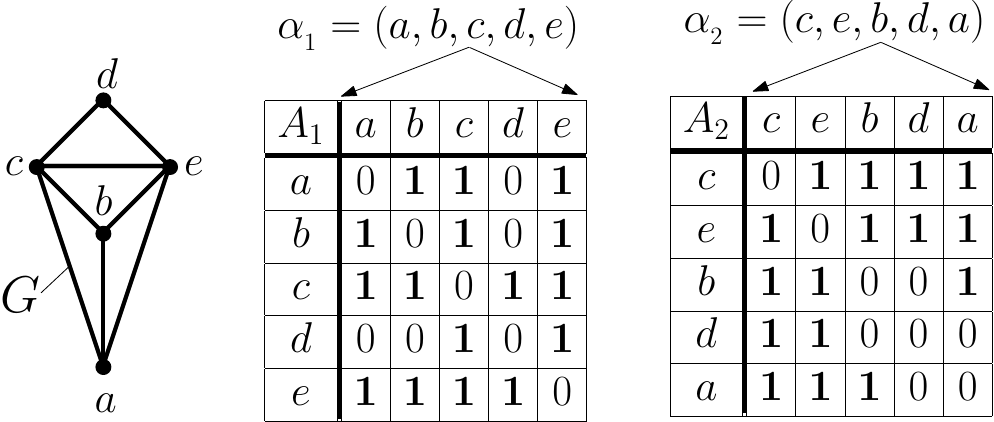} 
  \end{center}
  \vspace{-7pt}
  \caption{An example of two different adjacency matrices $A_{1}$ and $A_{2}$ for one graph $G$.}
 \label{nazarov:gr_table}
  \vspace{-5pt}
\end{figure}

\begin{definition}
	\textit{The code of the adjacency matrix} $A$ of a finite graph $G$ is the number $\mu(A)$ that is obtained as a result of converting the adjacency matrix into a binary number $\mu(A) = A(n,n) + A(n,n-1)\cdot 2 + ... + A(n,1)\cdot 2^{n-1} + A(n-1,n)\cdot 2^{n} + ... + A(1,1)\cdot 2^{n^{2}-1}$.
\end{definition}

\begin{definition}
The largest of all possible codes $\mu(A)$ of the adjacency matrices $A(\alpha,G)$ of a graph $G$ is called the \textit{maxi-code} of the graph $\mu_{\max}(G)$. If for some order of vertices $\alpha$ the code of the adjacency matrix $\mu(A) = \mu_{\max}(G)$, then we say that $\alpha$ \textit{corresponds to} the maxi-code. 
\end{definition}

In the work \cite{Nazarov_Graph1} the following statement about the maxi-code $\mu_{\max}(G)$ was proved.

\begin{statement} \label{nazarov:utv1}
{\it If two vertex orders $\alpha_{1} = (v^{1}_{1},...,v^{1}_{n})$ and $\alpha_{2} = (v^{2}_{1},...,v^{2}_{n})$ correspond to the maxi-code $\mu_{\max}(G)$ of the graph $G$, then the vertices in these orders will be pairwise automorphic   $v^{1}_{i} \sim v^{2}_{i}$ for all $i=\overline{1,n}$.}
\end{statement} 

For sequences of vertices $\alpha = (v_{1},...,v_{n})$ we use the standard index notation for elements $ v_{i} =\alpha(i)$.

\begin{definition}
Let the maxi-code of a graph $G$ correspond to some ordering $\alpha$ of vertices. Then we will call the \textit{number of the class of automorphism of vertices} $\overline{v}$ such natural number $N(\overline{v})$, which is equal to the first occurrence in the ordering $\alpha$ of a vertex from the class $\overline{v}$:
$ N(\overline{v}) = \! \min\limits_{i: \, \alpha(i) \in \overline{v} } \! \! i$. 
\end{definition}

\begin{definition}
Let a maxi-code of a graph $G$ correspond to some order of vertices $\alpha$. We call the \textbf{number of the symmetry class of edges} $(\overline{x,y})$ a natural number $N(\overline{x,y})$ equal to the first occurrence in the adjacency matrix $A(\alpha, G)$ of an edge from the class $(\overline{x,y})$ (taking into account the lexicographic order):
$N(\overline{x,y}) = \! \! \! \min\limits_{\substack{ A(i,j) = 1\\ (\alpha(i),\alpha(j))\sim (x,y)}} \! \! \!  \!   j + (i-1) \cdot n $.
\end{definition}

\begin{corollary} \label{nazarov:sle1}
  $N(\overline{v})$  induce a linear ordering on the classes of automorphic vertices: $\overline{v_{1}} \leq  \overline{v_{2}} \Leftrightarrow N(\overline{v_{1}}) \leq  N(\overline{v_{2}})$.
\end{corollary}

\begin{proof}
By definition, all the $N(v)$ are natural numbers, and according to the statement \ref{nazarov:utv1}, the numbers coincide only at automorphic vertices $v \sim u$. From this we obtain that the set of numbers $N(\overline{v_{1}}),..., N(\overline{v_{k}})$ is linearly ordered, and consequently, the order $\overline{v_{1}} \leq  \overline{v_{2}} \Leftrightarrow N(\overline{v_{1}}) \leq  N(\overline{v_{2}})$ will be linear.
\end{proof}

\begin{corollary} \label{nazarov:sle2}
	$N(\overline{v,u})$ induce a \textit{linear order} on the  automorphic edges: $(\overline{v_{1},u_{1}}) \leq  (\overline{v_{2},u_{2}}) \Leftrightarrow N(\overline{v_{1},u_{1}}) \leq  N(\overline{v_{2},u_{2}})$.
\end{corollary}
\begin{proof}
The proof can be carried out by analogy with the proof of Corollary \ref{nazarov:sle1}. If two vertex orders $\alpha_{1}$ and $\alpha_{2}$ correspond to the maxi-code of the graph $\mu_{\max}(G)$, then the adjacency matrices for these two orders coincide: $A(G,\alpha_{1}) = A(G,\alpha_{2})$. As a result, the edges that correspond to the same positions in these two matrices will be automorphic.
\end{proof}

\begin{definition} 
For the graph $G$, we introduce \textit{indices of the automorphism classes of vertices} $I(\overline{v})$ through the class numbers $N(\overline{v})$ using the following inductive rules:
\begin{enumerate}
	\item $N(\overline{v}) = \min\limits_{\overline{u}} N(\overline{u}) \, \Rightarrow \, I(\overline{v}) = 1$~--- basis of induction for determining the first index; 
	\item $ \left. \begin{matrix} 
			& (N(\overline{v}) > N(\overline{u})) \,  \wedge  \\ 
			& (\forall u^{*}\! \!\neq \!\! u \, \, N(\overline{v}) > N(\overline{u^{*}}) \Rightarrow N(\overline{u}) > N(\overline{u^{*}}))  \end{matrix} \right\rbrace \Rightarrow  I(\overline{v}) = I(\overline{u}) + 1 $~--- inductive step.
\end{enumerate} 
\end{definition}

\begin{definition} 
Following an analogy with the definition for vertices, we introduce the concept of the \textit{indices of the automorphism classes of  edges} $I(\overline{x,y})$, up to replacing the vertex classes with the edge automorphism classes.
\end{definition}

\begin{theorem} \label{nazarov:teo1}
If $G \cong H$, then for $u\in V(G)$, $v\in V(H)$ we have $I(\overline{u})=I(\overline{v}) $ if and only if there exists an isomorphism $ \phi: V(G) \rightarrow V(H) $ such that $\phi(u) = v$.
\end{theorem}
\begin{proof}
Let us first prove that from $\phi(u) = v $, where $\phi$ is an isomorphism, it follows that the indices of the automorphism classes $I(\overline{u})=I(\overline{v})$ coincide.\\
If two graphs are isomorphic, then they will have the same maxi-code $\mu_{\max}(G) = \mu_{\max}(H)$. In this case, the isomorphism can be represented as $\phi= \left(
	\begin{matrix}
  u_{1}\, ... \, u_{n} \\
  v_{1}\, ... \, v_{n} 
 \end{matrix} \right)$, where $\alpha_{1} = (u_{1},\, ... \, , u_{n}) $ is some order that corresponds to the maxi-code in $G$, and $\alpha_{2}=(v_{1},\, ... \, , v_{n})$ is the order that corresponds to the maxi-code in $H$. 
Since $N(\overline{u})$ is a natural number, then $\alpha_{1}(N(\overline{u}))$ will define some vertex $u^{\star} = \alpha_{1}(N(\overline{u}))$. This vertex $u^{\star}$ has two properties: $u \sim u^{\star}$ and, among all vertices automorphic with respect to $u$, $u^{\star}$ has the smallest index in $\alpha_{1}$.
If we assume that $N(\overline{u})\neq N(\overline{v})$, then we get $\phi(u^{\star}) \nsim \phi(u)$. 
As a result, we obtain a contradiction: $\phi(u^{\star}) \nsim \phi(u)$ and $u \sim u^{\star} $. Thus, for all isomorphic vertices ($\phi(u) = v $) the automorphism class numbers necessarily coincide $N(\overline{u}) = N(\overline{v})$. Using the corollary of \ref{nazarov:sle1} and taking into account the definition for the indices $I(\overline{v})$, we obtain $I(\overline{u}) = I(\overline{v})$. 
Now suppose that $G \cong H$ and the indices $I(\overline{u})=I(\overline{v}) $ coincide. Since $G \cong H$, there exists some isomorphism $\phi_{0}: V(G) \rightarrow V(H)$. For the image $ v^{*} =\phi_{0}(u)$, using the first part of this statement, we obtain $I(\overline{v}^{*})=I(\overline{u})=I(\overline{v})$.   
From this we can conclude that the vertices $v, v^{*}$ are automorphic:  $v \sim v^{*}$. We denote the automorphism that maps these vertices to each other in $\psi: \psi(v^{*}) = v$. Then the desired isomorphism $\phi$ can be defined as $\phi = \phi_{0} \circ \psi$. Since we defined the resulting mapping $\phi$ as a composition of two isomorphisms, the mapping $\phi$ itself will also be an isomorphism of graphs.
\end{proof}

\begin{theorem} \label{nazarov:teo2}
If $G\! \cong\! H$, then for  $(u_{1},u_{2})\in E(G)$ and $(v_{1},v_{2})\in E(H)$ we have $I(\overline{u_{1},u_{2}})=I(\overline{v_{1},v_{2}})$ if and only if there exists an isomorphism $ \phi: V(G) \rightarrow V(H)$, that $\phi(u_{1}) = v_{1}$ and $\phi(u_{2}) = v_{2}$, or $\phi(u_{1}) = v_{2}$ and $\phi(u_{2}) = v_{1}$.
\end{theorem}

\begin{proof}
Sufficiency can be proved in exactly the same way as for vertices. If we assume that there exists an isomorphism $ \phi: V(G) \rightarrow V(H) $ such that $\phi(u_{1}) = v_{1}$ and $\phi(u_{2}) = v_{2}$, then the numbers $N(\overline{u_{1},u_{2}})=N(\overline{v_{1},v_{2}}) $ coincide. And from the coincidence of the numbers for all edge classes it follows $I(\overline{u_{1},u_{2}})=I(\overline{v_{1},v_{2}}) $. 

Necessity can also be proved by analogy with Theorem \ref{nazarov:teo1}. Let $G\! \cong\! H$ and $I(\overline{u_{1},u_{2}})=I(\overline{v_{1},v_{2}}) $. For the isomorphism images $v_{1}^{*} =\phi_{0}(u_{1})$ and $v_{2}^{*} =\phi_{0}(u_{2})$, using the first part of this statement, we obtain $I(\overline{u_{1},u_{2}})=I(\overline{v_{1},v_{2}})= I(\overline{v^{*}_{1},v^{*}_{2}})$.
From this we can conclude that the edges are automorphic $(v_{1},v_{2}) \sim (v^{*}_{1},v^{*}_{2})$. We denote the automorphism that maps these vertices to each other in $\psi: \forall i \,\, \psi(v_{i}^{*}) = v_{i} $. The desired isomorphism is the composition $\phi = \phi_{0} \circ \psi$.
\end{proof}

An immediate consequence of Theorems \ref{nazarov:teo1} and \ref{nazarov:teo2} is the fact that the indices $I(\overline{v})$ and $I(\overline{v_{1},v_{2}})$ are unique identifiers for automorphism classes. Thus, we can use these indices instead of the vertices and edges themselves when defining an invariant for the class of isomorphic graphs $[G]$.
%%%
In particular, we can first represent the graph $G$ as a string of symbols using an algorithm similar to the algorithm for constructing a linear notation for molecular graphs: SMILES (see   \cite{Weininger1989}). Then, in this string, we can replace the vertices $v$ with their indices $I(\overline{v})$ and obtain an invariant for the graph $G$.

\begin{definition} \label{nazarov:defSymNot}
\textit{Symmetric linear notation} $ \mathfrak{L} (G)$ for a connected graph $G$~--- is a string of symbols that is determined based on the following algorithm:\\
\underline{Step 0}: Introduce the set of excluded vertices $V_{E}$, as well as the set of temporarily replaced vertices $V_{R}$.
At the beginning of the algorithm, both of these sets are assumed to be empty $V_{R}=V_{E} = \varnothing$. For each new vertex $v$ placed in the set of replaced $V_{R}$, we will assign a special code $\# m$, where $m$~--- is the number of elements $|V_{R}|$ after the addition of this vertex $v$. The end of the algorithm will be indicated by the placement of all vertices $V(G)$ of graph $G$ in the set of excluded vertices $V_{E}$. \\
\underline{Step 1}: We choose any vertex $v_{1}$ of the first index ($I(\overline{v}_{1})=1$) and write it in the $\mathfrak{L} (G)$, also adding there the first opening parenthesis $ \mathfrak{L} (G) = v_{1} \Big[_{1} \ldots $ We place $v_{1}$ in the set of replaced $V_{R}$ and  assign it the first special code $\# 1$.\\ 
\underline{Step 2} (recursive): At the beginning of this step, we know that the linear notation has been extended to some vertex $v_{k}$. Moreover, some of the considered vertices $\left\lbrace v_{i} \right\rbrace_{i=1}^{i=k}$ were placed in the set of excluded $V_{E}$, and the other part in the set of replaced $V_{R}$. Instead of the replaced vertices, we will use the codes $\# 1, \ldots, \# m$ (for example, code $\# m $ corresponds to vertex $v_{k}$). Among all the vertices from the neighbourhood of $v_{k}$ that were not removed or replaced in the previous steps, we choose any $v_{k+1}$ that satisfies the following conditions:
\begin{enumerate}
	\item The highest priority is given to the vertex with the minimum index $I(v_{k+1})$.
	\item If two vertices have the same indices $I(v_{k+1}) = I(v^{*}_{k+1})$, then preference will be given to the vertex $v_{k+1}$ that has a lower edge index $I(v_{k},v_{k+1}) < I(v_{k},v^{*}_{k+1})$. 
	\item If the indices of vertices and edges of $v_{k+1}$ and $v^{*}_{k+1}$ coincide, then preference will be given to the vertex that is closer on the graph $G$ to the vertex of the replacement code $\#1$ (can be connected to it by a path of a minimum number of vertices). 
If these distances coincide for vertices $v_{k+1}$ and $v^{*}_{k+1}$~--- then we will choose the vertex, which is closer on graph $G$ to the  the vertex of the code $\#2$, and so on up to code $\# (m-1)$.  
	\item If the vertices $v_{k+1}$ and $v^{*}_{k+1}$ from the environment $v_{k}$ could not be distinguished using three conditions on the indices of the vertices, edges, and distances to the vertices of the codes $\# 1, \ldots, \# m$, then any one of them can be chosen.
\end{enumerate}

If the vertex $v_{k+1}$ can be found using these four criteria, then we add it to the linear notation $\mathfrak{L} (G) = \ldots \xrightarrow{i_{k}} v_{k} \Big[_{m} \! \! \ldots \, ; \xrightarrow{i_{k+1}} v_{k+1} \Big[_{m+1} \! \! \ldots $ and to the set $V_{R}$ of the temporarily replaced   vertices. 
After this, we proceed to recursive step 3 for vertex $v_{k+1}$. If $v_{k+1}$ cannot be found, we remove $v_{k}$ from the set $V_{R}$ and add it to the excluded set $V_{E}$. We then close the level $m$ bracket in the linear notation $\mathfrak{L}(G)$ and repeat step 2 for the vertex $v_{j}$ that currently corresponds to the code $\# (m-1)$.\\
\underline{Step 3} (recursive): At the beginning of this step  we know that the linear notation has been extended to some vertex $v_{k}$. Also, some of the considered vertices $\left\lbrace v_{i} \right\rbrace_{i=1}^{i=k}$ were placed in the set of excluded $V_{E}$, and the other part in the set of replaced $V_{R}$. 
Instead of the replaced vertices, we will use the codes $\# 1, \ldots, \# m$ (the code $\# m $ corresponds to $v_{k}$). First, we check whether there is an edge in the graph that connects the vertex $v_{k}$ and the vertex of code $\# 1$. 
If such an edge exists and it has not yet been considered in the notation, then we add\footnote{The number $i=I(v,v_{k})$~--- is the automorphism index of this edge.} it to 
notation: $\mathfrak{L} (G) = \ldots \xrightarrow{i_{k}} v_{k} \Big[_{m} \! \! \xrightarrow{i} \# 1 ; \ldots \, $ Then we repeat the process in strict order for code $\# 2$, and so on until the last code $\# (m -1)$ is checked. Upon completion of this process, we proceed to recursive step 2 for the vertex $v_{k}$. 
\end{definition}

\begin{definition} 
\textit{A linear notation of a class of isomorphic graphs} $[G]$ is a string $I[G]$ that is obtained from any arbitrary symmetric linear notation $ \mathfrak{L} (G)$ by replacing all vertices $v$ with indices of vertex classes $I(\overline{v})$.
\end{definition}

\begin{example}   
Let us consider the operation of the algorithm for constructing $\mathfrak{L}(G)$ and the transition to the linear notation $I[G]$ using the example of the figure \ref{nazarov:gr_notation} with a step-by-step derivation of the sets $V_{R}$ and $V_{E}$.
\end{example}
\begin{figure}[htb] 
  \vspace{-3pt}
  \begin{center}
     \includegraphics[width= 1.0\textwidth]{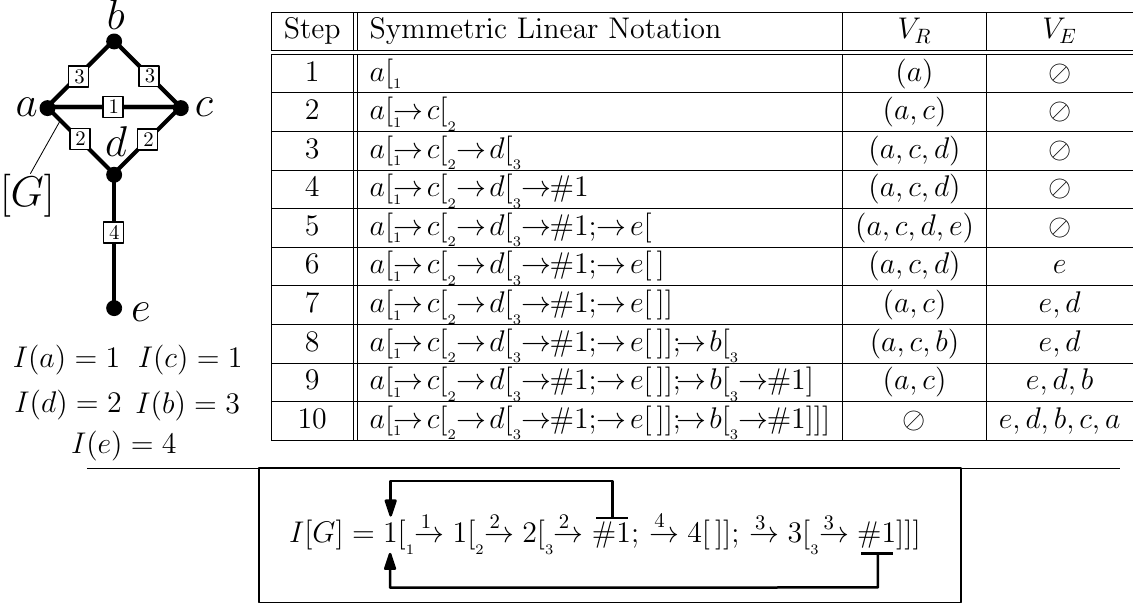} 
  \end{center}
  \vspace{-7pt}
  \caption{Step-by-step example of a linear notation $I[G]$ construction.}
 \label{nazarov:gr_notation}
   \vspace{-5pt}
\end{figure}

According to the Definition \ref{nazarov:defSymNot}, a single graph $G$ can have several different symmetric notations $\mathfrak{L} (G)$. In fact, for this to happen, it suffices that $G$ has at least
two automorphic\footnote{For example, if $G$ has two vertices of first index $I(v)=1$, then the notation $\mathfrak{L} (G)$ can begin with any of these vertices.} vertices. Thus, it is necessary to prove that the linear notation $I[G]$ is well defined and that it is indeed independent of the choice of $\mathfrak{L} (G)$.

\begin{statement}\label{nazarov:utv3}
	 The linear notation $I[G]$ is defined uniquely for any finite connected graph $G$.
\end{statement} 
\begin{proof}
Let us prove the uniqueness of the linear notation $I[G]$ by induction.\\
\underline{Base of induction}:
Consider step 1 of the algorithm for determining $\mathfrak{L} (G)$. After replacing the vertex $v_{1}$ with its index $I(v_{1})=1$, we have the same initial row $ I[G] = 1 \Big[_{1} \ldots $ for any $ \mathfrak{L} (G)$. As a result, we find that at step 1, ambiguity in the definition of $I[G]$ cannot arise, and the induction basis has been proven.\\
\underline{Inductive step}: Suppose that the ambiguity in the definition of $I[G]$ did not arise until a separate iteration of a recursive step 2 of the algorithm for determining $\mathfrak{L} (G)$. Then, after replacing the vertices $v_{k}$ and $v_{k+1}$ with their indices $I(v_{k})$ and $I(v_{k+1})$, we obtain the same row $I[G] = \ldots \xrightarrow{i_{k}} I(v_{k}) \Big[_{m} \! \! \ldots \, ; \xrightarrow{i_{k+1}} I(v_{k+1}) \Big[_{m+1} \! \! \ldots $ for any $ \mathfrak{L} (G)$. As a result, we obtain that at the recursive step 2, ambiguity in the definition of $I[G]$ cannot arise.\\
Suppose that the ambiguity in the definition of $I[G]$ did not arise until a particular iteration of a recursive step 3 of the algorithm for determining $\mathfrak{L} (G)$. If we assume that the codes $\# j$ associated with $v_{k}$ differ for two different $\mathfrak{L} (G)$, then we immediately obtain a contradiction with condition 3 for choosing new vertices in recursive step 2. As a result, the transition to $I[G] = \ldots \xrightarrow{i_{k}} I(v_{k}) \Big[_{m} \! \! \xrightarrow{i} \# 1 ; \ldots \,$ in step 3 will also be unambiguous.
\end{proof}

\begin{remark} \label{nazarov:remark1}
Note that for the linear notation $I[G]$, we can introduce the concepts of abstract vertices and abstract edges, associating them with the corresponding automorphism indices. Moreover, $I[G]$ can have multiple instances of abstract vertices that correspond to the same index $j$, and the same applies to abstract edges. 
\end{remark}

\begin{definition} \label{nazarov:Lin_not_colouring}
A \textit{colouring} \textit{for a linear notation} $I[G]$ of a single abstract \textit{vertex} $j$ \textit{in colour} $\alpha$ is a string $I^{*}_{\left\langle j, \alpha\right\rangle}[G]$ that is obtained from the linear notation $I[G]$ by replacing the first occurrence of $j$ with the pair $\left\langle j,\alpha\right\rangle $. A colouring for a larger number of vertices and a larger number of colours can be defined inductively as a sequential colouring of the graph linear notation by one single abstract vertex in each step. In this case, it suffices to require that at each colouring step we select only those vertices that were not coloured in the previous ones.
\end{definition}

\begin{theorem}   \label{nazarov:teo3}
	  The linear notation of the class of isomorphic graphs $I[G]$ is a complete graph invariant for any finite connected graph $G$. 
\end{theorem} 
\begin{proof}
Let us prove the necessity, that the existence of an isomorphism of the graphs $G_{1}\cong G_{2}$ will imply that $I[G_{1}] = I[G_{2}]$. First of all, using theorems \ref{nazarov:teo1} and \ref{nazarov:teo2}, we can guarantee that the graphs $G_{1}$ and $G_{2}$ have the same set of indices for vertices and edges. As a consequence, given the statement \ref{nazarov:utv3}, the notations $\mathfrak{L} (G_{1})$ and $\mathfrak{L} (G_{2})$ will differ only in the names of the vertices. As a result, we obtain that the linear notations coincide $I[G_{1}] = I[G_{2}]$.

Let us prove the sufficiency, that the coincidence of the linear notations $I[G_{1}] = I[G_{2}]$ implies an isomorphism of the graphs $ G_{1} \cong G_{2}$. To do this, we will use the colouring of $I[G]]$ with vertices from $G_{1} = (V_{1}, E_{1})$  (colours will be from  $V_{1}$). In doing so, we will choose the variant of the colouring $\mathfrak{L} (G_{1})$ that will uniquely correspond to the graph $G_{1}$ by the connections $E_{1}$ between the vertices, potentially enumerating all $n!$ possible variants. Then we will repeat a similar colouring procedure for $G_{2}$ and obtain $\mathfrak{L} (G_{2})$ for it.
Let us construct an explicit isomorphism by mapping all consecutive vertices of $\mathfrak{L} (G_{1})$ to the corresponding vertices of $\mathfrak{L} (G_{2})$. If we assume that the isomorphism condition is not satisfied for this mapping, then we immediately obtain a contradiction. Indeed, a violation of the isomorphism condition would mean that the codes $\# 1, \, ... \, , \#m$ occupy different positions in $\mathfrak{L} (G_{1})$ and $\mathfrak{L} (G_{2})$, which is impossible due to the identity $I[G_{1}] = I[G_{2}]$ and the unique definition of $I[G]$ based on $\mathfrak{L} (G)$. 
\end{proof}

To demonstrate the presented results, the figure \ref{nazarov:gr_examples} shows examples of $I[G]$ for three graphs $G_{1}$, $G_{2}$, and $G_{3}$, whose vertices are labelled with indices $I(\overline{v})$, and edges $I(\overline{u,v})$.
\begin{figure}[htb] 
  \vspace{-3pt}
  \begin{center}
     \includegraphics[width= 0.75\textwidth]{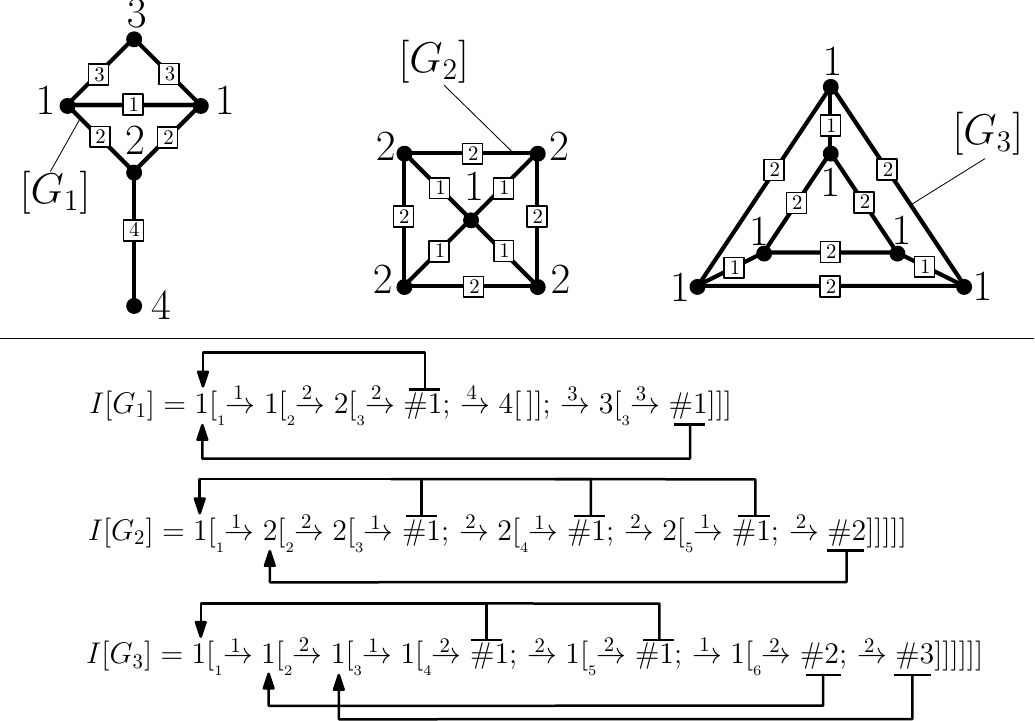} 
  \end{center}
  \vspace{-7pt}
  \caption{Examples of a linear notations $I[G]$ for the classes of isomorphic graphs $[G_{1}]$, $[G_{2}]$ and $[G_{3}]$. }
 \label{nazarov:gr_examples}
   \vspace{-5pt}
\end{figure} 
%

%-----------------------------------------------------------------------
\section{Practical application of linear notations}

The classical method for applying graph invariants mostly reduces to calculating and comparing the invariants each time two graphs are tested for isomorphism. Such a scheme would be quite impractical with a complete invariants, since constructing them is an ISO-hard problem. In particular, for the notations $I[G]$, we have the following complexity asymptotic (parameter $n=|V(G)|$):
\begin{enumerate}
	\item The transition from $G$ to $I[G]$ is a problem equivalent to the graph isomorphism problem (ISO); 
	\item Checking for a match for $I[G]$ is an $O(n^{2})$ complexity problem; 
	\item The transition from $I[G]$ to $G$ (based on vertex  colouring) is a problem of complexity $O(n^{2})$.
\end{enumerate}

The general idea is to store $I[G]$ in computer memory instead of the common graph $G$. This approach will radically simplify the isomorphism test, since it will not be necessary to pre-compute the invariant $I[G]$, and the procedure itself will be reduced to the polynomial problem of checking for the coincidence $I[G_{1}] = I[G_{2}]$.
For the graphs with a small number of vertices, we can use the basic algorithm to get from $G$ to $I[G]$, despite its belonging to the $\mathrm{ISO}$ class, since it only needs to be applied once before preserving the linear notation of $I[G]$. However, for graphs with an arbitrarily large number of vertices, this approach is substantially less practical.
In fact, efficiently constructing linear notations of $I[G]$ for large graphs requires defining standard classes of graphs directly based on $I[G]$ notation, and also requires introducing some set of polynomial-time operations on $I[G]$ to obtain more complex classes of isomorphic graphs.

First, we will consider how the following standard classes of graphs can be described using the linear notations $I[G]$: paths $P_{n}$, cycles $C_{n}$, as well as complete graphs $K_{n}$ and complete bipartite graphs $K_{n\, m}$ (for more details on these classes, see \cite{Zykov1987} and \cite{Belov1976}).

\begin{example} 
For paths $P_{n}$ we will consider two cases: when $n$ is an even number $n=2k$, and when the number $n$ is odd $n=2k-1$.
\[
	I[P_{2k}] = 1 \Big[_{1} \xrightarrow{1} 1 \Big[_{2} \xrightarrow{2} 2 \Big[_{3}  \ldots \xrightarrow{n} n [\, ] \Big] \ldots \Big] ;  \xrightarrow{2} 2 \Big[_{2} \xrightarrow{3} 3 \Big[_{3} \ldots \xrightarrow{n} n [\, ] \Big] \ldots \Big] \, \Big]
\]
\[
	I[P_{2k-1}] = 1 \Big[_{1} \xrightarrow{1} 2 \Big[_{2} \xrightarrow{2} 3 \Big[_{3} \ldots \xrightarrow{n-1} n [\, ] \Big] \ldots \Big] ; \xrightarrow{1} 2 \Big[_{2} \xrightarrow{2} 3 \Big[_{3} \ldots \xrightarrow{n-1} n [\, ] \Big] \ldots \Big] \, \Big]
\]
\end{example}

\begin{example} 
The linear notation for cycle graphs $C_{n}$ is defined with only one general pattern:
\[
	I[C_{n}] = 1 \Big[_{1} \xrightarrow{1} 1 \Big[_{2} \ldots \xrightarrow{1} 1 \Big[_{n}  \xrightarrow{1} \# 1 \Big] \ldots \Big] \, \Big]
\]
\end{example}

\begin{example}  
The linear notation for the complete graph $K_{n}$ is also uniquely defined:
	\[
	I[K_{n}] = 1 \Big[_{1} \xrightarrow{1} 1 \Big[_{2} \xrightarrow{1}  1 \Big[_{3} \xrightarrow{1} \#1 ; \xrightarrow{1}1\Big[_{4} \ldots \xrightarrow{1} 1 \Big[_{n}  \xrightarrow{1} \# 1 ; \xrightarrow{1} \# 2; \ldots ; \xrightarrow{1} \# n\!-\! 2 \Big] \ldots \Big]\, \Big]\, \Big] \, \Big]
\]
\end{example}

\begin{example}  
For bipartite graphs $K_{n \, m}$ we will consider two cases: when the numbers $n<m$, and also when $n = m$ (we will omit the third case $n>m$).
\begin{align*}
	I[K_{n \, m}] =  1 \Big[_{1}\!\! \ldots \xrightarrow{1}\!  1 & \Big[_{2m\! -\! 1} \! \!\! \! \xrightarrow{1} \# 2; \ldots ; \xrightarrow{1} \# 2m\! -\! 4; 
	 \xrightarrow{1} 2 \Big[_{2m} \! \! \xrightarrow{1} \# 1; \xrightarrow{1} \# 3; \ldots ; \xrightarrow{1} \# 2m\! -\! 3 \Big] ;\\ 
	 & \ldots ;  \xrightarrow{1} 2 \Big[_{2m}\! \! \! \xrightarrow{1} \# 1; \xrightarrow{1} \# 3; \ldots ; \xrightarrow{1} \# 2m\! -\! 3 \Big] \, \Big] \ldots \Big]
\end{align*}	
\[
	I[K_{n\, n}] = 1 \Big[_{1} \ldots \xrightarrow{1} \!  1 \Big[_{2n\!-\! 1} \! \! \!\! \! \xrightarrow{1} \# 2; \ldots ; \xrightarrow{1} \# 2n\! -\! 4 ; \xrightarrow{1} 1 \Big[_{2n} \! \! \xrightarrow{1} \# 1; \xrightarrow{1} \# 3; \ldots ; \xrightarrow{1} \# 2n\! -\! 3 \Big]   \Big] \ldots \Big]  
\]
\end{example}

\begin{remark}\label{nazarov:remark2}
As can be seen from the examples for $I[P_{2k-1}]$ and $I[K_{n\, m}]$, the linear notation $I[G]$ can be easily adapted for data compression if we use the special $S\times n$ notation for substrings that are repeated $n$ times in a row. In the case of $I[P_{2k-1}]$, we would get reduced notation: $I^{*}[P_{2k-1}] = 1 \Big[_{1} \xrightarrow{1} 2 \Big[_{2} \xrightarrow{2} 3 \Big[_{3} \ldots \xrightarrow{n-1} n [\, ] \Big] \ldots \Big]; \times 2 \Big]$.
\end{remark}

From examples of defining standard graph classes, we will move on to a construction of some elementary graph operations. In this case, we will describe an arbitrary operation on classes of isomorphic graphs in the form of a string processing algorithm for linear notations $I[G]$.

\begin{example} 
We introduce the operation of chain extension by defining the resulting $ \tilde {G}$ graph as obtained from the original $G$ by adding to all the vertices of the first degree\footnote{The degree of a vertex $\mathrm{deg} (v)$~ is the number of vertices that are connected to $v$ by edges $\mathrm{deg} (v) = | \left\lbrace v^{*} | \, (v,v^{*}) \in E \right\rbrace|$.}
exactly one new vertex (implementation example see Fig. \ref{nazarov:gr_elong}).

\begin{figure}[htb] 
  \vspace{-3pt}
  \begin{center}
     \includegraphics[width= 0.85\textwidth]{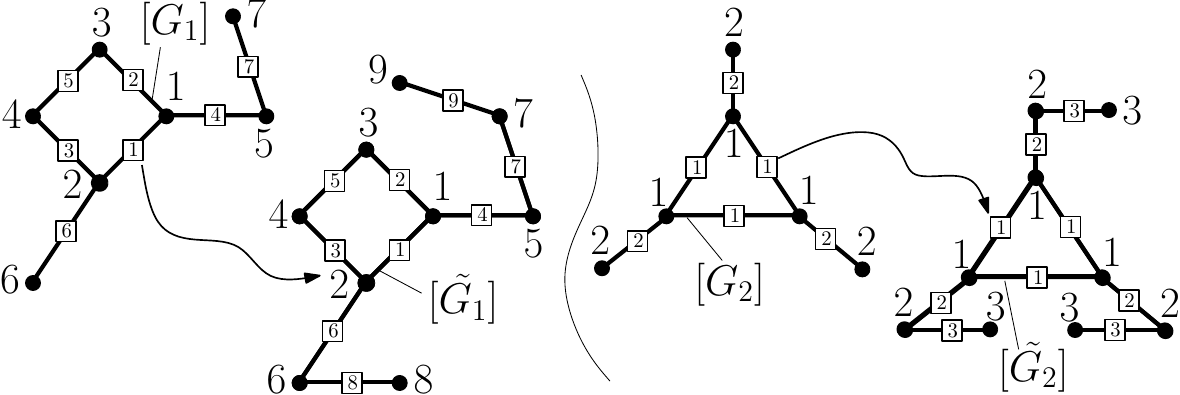} 
  \end{center}
  \vspace{-7pt}
  \caption{ An example of using the chain extension operation. }
 \label{nazarov:gr_elong}
   \vspace{-5pt}
\end{figure} 

The algorithm for constructing a linear notation $I[\tilde{G}]$ for the resulting graph based on the original notation $I[G]$ can be reduced to the following three steps of string processing.\\
\underline{Step 1}: Search the string $I[G]$ for vertices that match the pattern ($\ldots \xrightarrow{j} k [\,] \ldots$) for a first-degree vertex. Among all these vertices, calculate the minimum automorphism index $k$ and the maximum $k+m$ that match this pattern.\\
\underline{Step 2}: Introduce new class indices $k\!+\!m\!+\!1,\, k\!+\!m\!+\!2, \ldots , k\!+\!2m$, which will correspond to new vertices of  $\tilde{G}$.\\
\underline{Step 3}: Obtain $I[\tilde{G}]$ from the original linear notation $I[G]$ by substituting substrings according to a pattern: replace the string $\left( \ldots \xrightarrow{j+ i} k\!+\! i [\,] \ldots \right)$ with $\left( \ldots \xrightarrow{j+ i} k\! +\! i \Big[\xrightarrow{j+ i+ m} k\!+\!m\!+\! i [\,] \Big] \ldots\right)$.
\end{example}

\begin{example} 
Let's consider the operation of appending vertices of the first degree to vertices of the first automorphism class (for an implementation of the operation for specific graphs, see figure \ref{nazarov:gr_concat}). Formally, when implementing this operation we add exactly one vertex $v^{*}$ of degree $\mathrm{deg}(v^{*}) =1$ to each vertex $v$ with automorphism index $I(v)=1$.
\begin{figure}[htb] 
  \vspace{-3pt}
  \begin{center}
     \includegraphics[width= 0.85\textwidth]{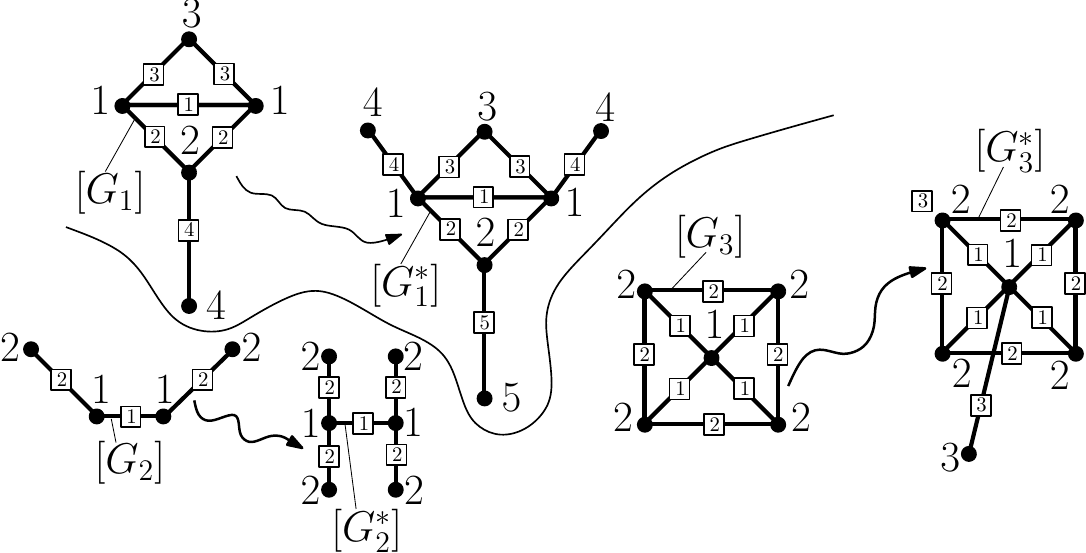} 
  \end{center}
  \vspace{-7pt}
  \caption{ An example of using the operation of appending a new vertex to vertices of the first class. }
 \label{nazarov:gr_concat}
   \vspace{-5pt}
\end{figure} 

The algorithm for constructing a linear notation $I[{G}^{*}]$ for the resulting graph based on the original notation $I[G]$ can be reduced to the following three steps.\\
\underline{Step 1}: Search the string $I[G]$ for vertices that match the pattern ($\ldots \xrightarrow{j} k [\,] \ldots$) and are connected by an edge to a first-class vertex. If we do not find such vertices, then we go to Step 2 of this algorithm. Otherwise, we add an additional vertex of class $k$ to the neighbourhood of each first-class vertex by performing a string substitution according to the pattern: replace the strings $\left( 1 \Big[ \ldots \xrightarrow{j} k [\,] \Big] \ldots \right)$ with $\left( 1 \Big[ \ldots \xrightarrow{j} k [\,]; \xrightarrow{j} k [\,] \Big] \ldots \right)$.\\
\underline{Step 2}: If in the previous step we were unable to find vertices of the first degree connected with vertices of the first automorphism class and complete the algorithm with the corresponding replacement, then we will simply search for vertices of the first degree ($\ldots \xrightarrow{j} k [\,] \ldots$). Among all these vertices, we calculate the minimum automorphism index $k$ and the maximum $k+m$.\\
\underline{Step 3}: We perform a string replacement of all vertex indices from $k$ to $k+m$ using the following pattern: replace the strings $\left( \ldots \xrightarrow{j+i} k\!+ \! i [\,] \ldots \right)$ with $\left( \ldots \xrightarrow{j+ i+ 1} k\! +\! i\!+\!1 [\,] \ldots\right)$. After this, we add one new vertex of class $k$ to the neighbourhood of each vertex of class 1 according to the pattern: $\left( 1 \Big[ \ldots \xrightarrow{j} k [\,] \Big] \ldots \right)$.

\end{example}

\begin{remark}
For all the operations discussed so far, as well as for the description of standard graph classes, we will omit the proof of correctness. If desired, these proofs can be easily obtained by analogy with those already discussed.
\end{remark}

In definition \ref{nazarov:Lin_not_colouring} we introduced the concept of colouring for individual abstract vertices of the linear notation $I[G]$. By exact analogy with this definition, we can introduce a colouring for abstract edges on $I[G]$. If the original graph $G$ had $n$ vertices, then by iteratively colouring of all the abstract vertices of $I[G]$ in $n$ different colors, we obtain $\mathfrak{L}(G^{*})$ and restore the original graph up to an isomorphism:  $G^{*} \cong G$. It is important to note that such an operation of transition from $I[G]$ to the ordinary graph $G^{*} \cong G$ in the limit case will require restoring the connections between every two vertices of the graph, and as a consequence will have maximum complexity $O(n^{2})$, where $n = |V(G)|$.
Besides reconstructing a regular graph based on $I[G]$, the colouring algorithm can also be used to construct molecular\footnote{For more information on structural chemical formulas and molecular graphs, see \cite{Brecher2006}.}
graphs, if the vertex colors are the names of chemical elements $H$, $O$, $C$, $N$, and the edge colours are the chemical bonds.

\begin{example} 
Figure \ref{nazarov:gr_molec} shows the molecular graph of pyridine, as well as the labelled line notation that was obtained for this graph using an iterative colouring algorithm for vertices and edges (double bonds are shown by a double arrow $\Rightarrow$).
\begin{figure}[htb] 
  \vspace{-3pt}
  \begin{center}
     \includegraphics[width= 1.0\textwidth]{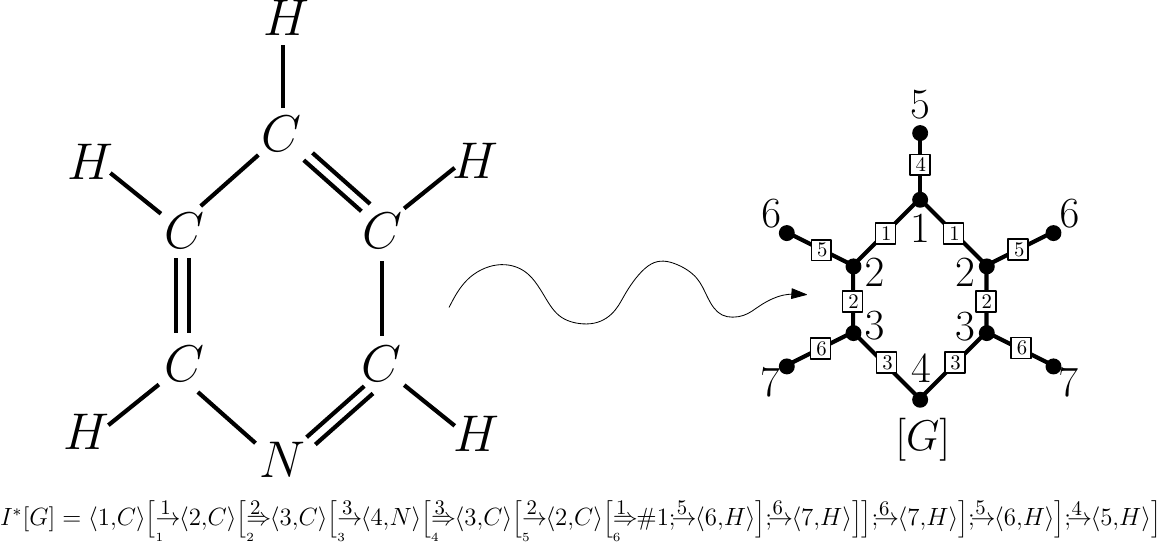} 
  \end{center}
  \vspace{-7pt}
  \caption{ An example representation of a molecular graph in the form of a labelled linear notation. }
 \label{nazarov:gr_molec}
   \vspace{-5pt}
\end{figure} 
By describing molecular graphs based on the labelled linear notations $I^{*}[G]$ we get a three obvious advantages:
\begin{itemize}
	\item	 Checking whether two graphs are isomorphic will be an $O(n^{2})$ problem.
	\item	 All possible symmetries of the molecule will be represented in a convenient form through indices of the automorphism classes of vertices and edges.
	\item	 Specifying reactions between molecules as operations on $I[G]$ notations will actually be reduced to string processing, which simplifies their implementation.
\end{itemize}
\end{example}

In addition to the applications already considered, the colouring algorithm for $I[G]$ can also be used to construct paths on a graph. In this case, the colouring will be performed by the natural numbers $1,2,\ldots , n$, and any two vertices coloured by adjacent numbers $i$, $i+1$, must be connected by an edge in $I[G]$. Note that of particular interest is the search for a Hamiltonian\footnote{A Hamiltonian path passes through all vertices of the graph exactly once.} path in the linear notation of $I[G]$ (for more details on these paths, see \cite{Harari2006}). Using as a basis an enumeration of graphs with $|V| \leqslant 7$ vertices, the following conjecture regarding the existence of a Hamiltonian path was put forward.

\begin{conjecture}
	 In order for a Hamiltonian path to exist on $I[G]$, it is necessary and sufficient that $I[G]$ satisfy one of the four possible patterns (notation $\# \# \# $ denotes a string of return codes with abstract edges $\rightarrow \# 1 \ldots \rightarrow \# m$): 
\begin{itemize}
	\item $I[G] = 1 \Big[_{1} \rightarrow k \Big[_{2} \ldots \Big] \, \Big]$
	\item $I[G] = 1 \Big[_{1} \rightarrow k_{1} \Big[_{2} \ldots \Big]; \rightarrow k_{2} \Big[_{2} \ldots \Big] \,  \Big]$
	\item $I[G] = 1 \Big[_{1} ... \rightarrow k_{1} \Big[_{i1} \#\#\# ; \rightarrow k_{2} \Big[_{i2} \rightarrow \#1  ... \Big] \,   \Big]; \rightarrow k_{3} \Big[_{i3} \ldots \Big] \,  \Big]$	
	\item $I[G] = 1 \Big[_{1} ... \rightarrow k_{1} \Big[_{i1} \#\#\# ; \rightarrow k_{2} \Big[_{i2} ...  \rightarrow k_{3} \Big[_{i3} \rightarrow \#1; \, \#\#\# \Big]  \,  \Big] \,   \Big]; \rightarrow k_{4} \Big[_{i4} \ldots \Big] \,  \Big]$
\end{itemize}	 
Inside the blocks $\Big[_{2}\ldots \Big]$, $\Big[_{i2}\ldots \Big]$, $\Big[_{i3}\ldots \Big]$ and $\Big[_{i4}\ldots \Big]$ the number of nested brackets must be exactly equal to the number of abstract vertices (no extra branching allowed).
\end{conjecture}

%-----------------------------------------------------------------------  

\section*{Conclusion}%--------------------------

In summary of the research conducted in this paper, it can be argued that the linear notation $I[G]$ may become a convenient alternative to classical graphs $G$ for a wide range of practical applications. Therefore, it seems relevant to continue research in this area and to develop new algorithms on $I[G]$, including algorithms for implementing graph operations.

Moreover, if we can construct a basis\footnote{By basis we mean a set of operations with which we can obtain any finite graph by applying their superposition to some simplest graphs.} of operations on $I[G]$ of polynomial complexity for some graph classes, then the $I[G]$ notation can be used in practice for storing large graphs of that classes.

It should be noted that the memory storage complexity of $I[G]$ is asymptotically equal to $O(n^{2})$, where $n$ is the number of graph vertices. The data compression ratio for $I[G]$ can be increased by using the technique presented in remark \ref{nazarov:remark2}, as well as by dropping the indices $m$ for all brackets $\Big[_{m} \ldots \Big]$ in the notation of $I[G]$. These indices were introduced in the definition of $\mathfrak{L}(G)$ as a syntactic sugar  to make the notation more visual.

In addition to the applications already considered for linear notations $I[G]$, we can try to construct an optimized algorithm for the complete enumeration of graphs with $n$ vertices, discarding all isomorphic copies. In fact, if we enumerate $I[G]$ rather than the adjacency matrices $A$ of graphs $G$, we can achieve a substantial increase in the performance of such an algorithm, since the number of classes of isomorphic graphs with $n$ vertices is significantly smaller than the number of tables of zeros and ones of size $n\times n$ (see, for example, \cite{HararyPalmer1973}).

As already noted in \cite{Nazarov_Graph1}, the linear notation $I[G]$ can be fairly easily generalized to the case of ordered graphs, as well as to hypergraphs (see \cite{Zykov1974}). On the other hand, alternative indices can be introduced for the automorphism classes of vertices $\overline{v}$ and edges $(\overline{u,v})$ by choosing a different complete numerical invariant of the adjacency table instead of the maxi-code. In particular, as noted in \cite{Nazarov_Graph1}, a mini-code can be used instead of the maxi-code. From a practical point of view, such alternative constructions will be appropriate if it can be shown that they simplify the definition of certain algorithms on classes of isomorphic graphs. For example, if they can be used to more easily define certain operations on classes of isomorphic graphs.

%-----------------------------------------------------------------------  

\bibliographystyle{unsrt}  


\end{document}